\documentclass[letterpaper]{article}
\usepackage[margin=1.00in]{geometry}
\usepackage{mathtools}
\usepackage{amsmath,amsfonts,amssymb,amsthm,url,verbatim}
\usepackage{graphicx,subcaption}
\usepackage{enumitem}
\usepackage{authblk}
\usepackage[ pdftex, plainpages = false, pdfpagelabels,
                 pdfpagelayout = useoutlines,
                 bookmarks,
                 bookmarksopen = true,
                 bookmarksnumbered = true,
                 breaklinks = true,
                 linktocpage=all,
                 pagebackref=false,
                 colorlinks = true,
                 linkcolor = BrickRed,
                 urlcolor  = blue,
                 citecolor = BrickRed,
                 anchorcolor = green,
                 hyperindex = true,
                 hyperfigures
                 ]{hyperref}
\usepackage[usenames, dvipsnames]{xcolor}

\DeclarePairedDelimiter\bra{\langle}{\rvert}

\DeclarePairedDelimiter\ket{\lvert}{\rangle}

\DeclarePairedDelimiterX\braket[2]{\langle}{\rangle}{#1 \delimsize\vert #2}
\DeclarePairedDelimiterX\bbraket[2]{\langle\!\langle}{\rangle\!\rangle}{#1 \delimsize| #2}
\DeclarePairedDelimiterX\cbraket[2]{(\!(}{)\!)}{#1 \delimsize\| #2}
\DeclarePairedDelimiterX\ketbra[2]{\lvert}{\rvert}{#1 \delimsize\rangle\!\langle #2}
\DeclarePairedDelimiterX\kketbra[2]{|}{|}{#1 \delimsize\rangle\!\rangle\!\langle\!\langle #2}
\DeclarePairedDelimiterX\cketbra[2]{\|}{\|}{#1 \delimsize)\!)\!(\!( #2}
\DeclarePairedDelimiterX\inner[2]{\langle}{\rangle}{#1,#2}

\DeclarePairedDelimiter\norm{\lVert}{\rVert}

\def\tr{{\rm tr}\,}
\newtheorem{theorem}{Theorem}
\newtheorem{lemma}{Lemma}
\newtheorem{corollary}{Corollary}
\newtheorem{example}{Example}
\newtheorem{conjecture}{Conjecture}

\begin{document}
\title{The Varieties of Minimal Tomographically Complete Measurements} \author[$\dag\star$]{John B.\ DeBrota}
\author[$\dag\star$]{Christopher A.\ Fuchs} \author[$\dag$]{Blake
  C.\ Stacey} \affil[$\dag$]{\small
  \href{http://www.physics.umb.edu/Research/QBism}{QBism Group},
  Physics Department, University of Massachusetts Boston, \par 100
  Morrissey Boulevard, Boston MA 02125, USA}
\affil[$\star$]{\href{http://stias.ac.za/events/workshop-on-participatory-realism/}{Stellenbosch
    Institute for Advanced Study} (STIAS), Wallenberg Research Center
  at Stellenbosch University, Marais Street, Stellenbosch 7600, South
  Africa} \date{\today}

\maketitle
\begin{abstract}
Minimal Informationally Complete quantum measurements, or MICs, 
illuminate the structure of quantum theory and how it departs from 
the classical. Central to this capacity is their role as 
tomographically complete measurements with the fewest possible 
number of outcomes for a given finite dimension. Despite their 
advantages, little is known about them. We establish general 
properties of MICs, explore constructions of several classes of 
them, and make some developments to the theory of MIC Gram matrices.
These Gram matrices turn out to be a rich subject of inquiry, 
relating linear algebra, number theory and probability. Among our 
results are some equivalent conditions for unbiased MICs,  a 
characterization of rank-1 MICs through the Hadamard product, several 
ways in which immediate properties of MICs capture the abandonment 
of classical phase space intuitions, and a numerical study of MIC 
Gram matrix spectra. We also present, to our knowledge, the first 
example of an unbiased rank-1 MIC which is not group covariant. This 
work provides further context to the discovery that the symmetric 
informationally complete quantum measurements (SICs) are in many 
ways optimal among MICs. In a deep sense, the ideal measurements of quantum physics are not orthogonal bases.
\end{abstract}
\section{Introduction}
\label{sec:intro}
\normalsize
A significant part of science is the pursuit of measurements that are as informative as possible.  Attempts to provide an elementary explanation of ``the scientific method'' sometimes convey the notion that an ideal measurement is one which is exactly reproducible, always yielding the same answer when applied in succession.  But this notion has fairly obvious problems, for example, when the system being measured is dynamical.  When the experiment's sought outcome is the position of Mars at
midnight, the numbers will not be the same from one night to the next, and yet Kepler could run a scientific revolution on that data.  A more refined standard would be that an ideal measurement is one that provides enough information to project the complete dynamical trajectory of the measured system through phase space. Quantum physics frustrates this ambition by denying the phase space: Quantum uncertainties are not uncertainties about the values of properties that pre-exist the act of
measurement. Yet the ideal of a sufficiently informative measurement, the expectations for which fully fix the expectations for any other, can still be translated from classical thought to quantum, and doing so illuminates the nature of quantum theory itself. 

Let $\mathcal{H}_d$ be a $d$-dimensional complex Hilbert
space, and let $\{E_i\}$ be a set of positive semidefinite operators
on that space which sum to the identity:
\begin{equation}
\sum_{i=1}^N E_i = I.
\end{equation}
The set $\{E_i\}$ is a \emph{positive-operator-valued measure} (POVM),
which is the mathematical representation of a measurement process in
quantum theory. Each element in the set --- called an \emph{effect}
--- stands for a possible outcome of the measurement~\cite[\S
  2.2.6]{Nielsen:2010}. A POVM is said to be \emph{informationally
  complete} (IC) if the operators $\{E_i\}$ span $\mathcal{L}(\mathcal{H}_d)$, the space of
Hermitian operators on $\mathcal{H}_d$, and an IC POVM is said to be
\emph{minimal} if it contains exactly $d^2$ elements.  For brevity, we
can call a minimal IC POVM a MIC.

A matrix which captures many important properties of a MIC is its Gram matrix, that is, the matrix $G$ whose entries are given by
\begin{equation}
[G]_{ij} := \tr E_i E_j\;.
\end{equation}
Of particular note among MICs are those which enjoy the symmetry
property
\begin{equation}
[G]_{ij} = [G_{\rm SIC}]_{ij} := \frac{1}{d^2} \frac{d\delta_{ij} + 1}{d+1}\;.
\end{equation}
These are known as \emph{symmetric} informationally complete POVMs, or
SICs for short~\cite{Zauner:1999, Renes:2004, Scott:2010a,
  Fuchs:2017a}. In addition to their purely mathematical properties,
SICs are of central interest to the technical side of QBism, a
research program in the foundations of quantum
mechanics~\cite{Fuchs:2014b, Fuchs:2013, Healey:2016,
  Fuchs:2016a}. Investigations motivated by foundational concerns led
to the discovery that SICs are in many ways optimal among
MICs~\cite{Appleby:2014b, Appleby:2015, DeBrota:2018}. In this paper,
we elaborate upon some of those results and explore the conceptual
context of MICs more broadly.

MICs provide a new way of understanding the Born Rule, a key step in
how one uses quantum physics to calculate probabilities.  The common
way of presenting the Born Rule suggests that it fixes probabilities
in terms of more fundamental quantities, namely quantum states and
measurement operators. MICs, however, suggest a change of
viewpoint. From this new perspective, the Born Rule should be thought
of as a \emph{consistency condition between the probabilities assigned
  in diverse scenarios} --- for instance, probabilities assigned to
the outcomes of complementary experiments. The bare axioms of
probability theory do not themselves impose relations between
probabilities given different conditionals: In the abstract, nothing
ties together $P(E|C_1)$ and $P(E|C_2)$. Classical intuition suggests
one way to fit together probability assignments for different
experiments, and quantum physics implies another. The discrepancy
between these standards encapsulates how quantum theory departs from
classical expectations~\cite{Fuchs:2017b, Stacey:2018b}. MICs provide
the key to addressing this discrepancy; any MIC may play the role of a
reference measurement through which the quantum consistency condition
may be understood. To understand MICs is to understand how quantum
probability is like, and differs from, classical.

In the next section, we introduce the fundamentals of quantum
    information theory and the necessary concepts from linear algebra to prove a few basic results about MICs and comment on their conceptual meaning. Among the results included are a characterization of unbiased MICs, a condition
in terms of matrix rank for when a set of vectors in~$\mathbb{C}^d$
can be fashioned into a MIC, and an explicit example of an unbiased MIC which is not group covariant. In Section~\ref{sec:constructions},
we show how to construct several classes of MICs explicitly and note some properties of their Gram matrices. In Section~\ref{sec:optimal}, we explore several ways in which SICs are optimal among MICs for the project of differentiating the quantum from the classical, a topic complementing one of our recent papers~\cite{DeBrota:2018}. To conclude, in Section~\ref{sec:numerics}, we conduct an initial numerical study of the Gram matrix eigenvalue spectra of randomly-chosen MICs of four
different types. The empirical
eigenvalue distributions we find have intriguing features, not all of
which have been explained yet. 

\section{Basic Properties of MICs}
\label{sec:basics}
We begin by briefly establishing the necessary notions from quantum
information theory on which this paper is grounded. In quantum
physics, each physical system is associated with a complex Hilbert
space. Often, in quantum information theory, the Hilbert space of
interest is taken to be finite-dimensional. We will denote the
dimension throughout by $d$. A \emph{quantum state} is a positive
semidefinite operator of unit trace. The extreme points in the space of quantum states are the rank-1
projection operators:
\begin{equation}
  \rho = \ketbra{\psi}{\psi}.
\end{equation}
These are idempotent operators; that is, they all satisfy $\rho^2 =
\rho$. If an experimenter ascribes the
quantum state $\rho$ to a system, then she finds her probability for
the $i^{\rm th}$ outcome of the measurement modeled by the POVM
$\{E_i\}$ via the Hilbert--Schmidt inner product:
\begin{equation}
  p(E_i) = \tr\rho E_i.
\end{equation}
This formula is a standard presentation of the Born Rule. The
condition that the $\{E_i\}$ sum to the identity ensures that the
resulting probabilities are properly normalized.

If the operators $\{E_i\}$ span the space of Hermitian
operators, then the operator $\rho$ can be reconstructed from its
inner products with them. In other words, the state $\rho$ can be
calculated from the probabilities $\{p(E_i)\}$, meaning that the
measurement is ``informationally complete'' and the state $\rho$ can,
in principle, be dispensed with. Any MIC can thus be considered a
``Bureau of Standards'' measurement, that is, a reference measurement
in terms of which all states and processes can be
understood~\cite{Fuchs:2002}. Writing a quantum state $\rho$ is often
thought of as specifying the ``preparation'' of a system, though this
terminology is overly restrictive, and the theory applies just as well
to physical systems that were not processed on a laboratory
workbench~\cite{Fuchs:2011c}.

Given any POVM $\{E_i\}$, we can always write its elements as
unit-trace positive semidefinite operators with appropriate scaling
factors we call \textit{weights}:
\begin{equation}
E_i := e_i \rho_i, \hbox{ where } e_i = \tr E_i.
\end{equation}
If the operators $\rho_i$ are all rank-1 projectors, we will refer to
the set $\{E_i\}$ as a \emph{rank-1 POVM}. We will call a POVM
\emph{unbiased} when the weights $e_i$ are all equal. Such operator sets
represent quantum measurements that have no intrinsic bias: Under the Born Rule they map
the ``garbage state'' $(1/d)I$ to a flat probability distribution. For an unbiased MIC, the
condition that the elements sum to the identity then fixes $e_i =
1/d$. 

A \textit{column (row) stochastic matrix} is a real matrix with nonnegative entries whose columns (rows) sum to $1$. If a matrix is both column and row stochastic we say it is \textit{doubly stochastic}. The following theorem allows us to identify an unbiased MIC from a glance at its Gram matrix or Gram matrix spectrum.

\begin{theorem}\label{unbiased}
    Let $\{E_i\}$ be a MIC and $\lambda_{\rm max}(G)$ be the maximal eigenvalue of its Gram matrix $G$. The following are equivalent:
    \begin{enumerate}
        \item $\{E_i\}$ is unbiased.
        \item $dG$ is doubly stochastic.
        \item $\lambda_{\rm max}(G)=1/d$.
    \end{enumerate}
\end{theorem}
\begin{proof}
  
    The equivalence of the first two conditions is readily shown. We show (2)$\iff$(3). Let $\ket{v}:=\frac{1}{d}(1,\ldots,1)^{\rm T}$ be the normalized $d^2$ element uniform vector of $1$s. If $dG$ is doubly stochastic, $\ket{v}$ is an eigenvector of $dG$ with eigenvalue $1$, and the Gershgorin disc theorem \cite{Leinster:2016} ensures $\lambda_{\rm max}(G)=1/d$. For any MIC,
    \begin{equation}\label{minmaxeval}
        \lambda_{\rm max}(G)\geq\bra{v}G\ket{v}=\frac{1}{d}\;,
    \end{equation}
    with equality iff $\ket{v}$ is an eigenvector of $G$ with eigenvalue $1/d$. Since $G\ket{v}=(e_1,\ldots,e_{d^2})^{\rm T}$, $\ket{v}$ is an eigenvector of $G$ iff $e_i=1/d$ for all $i$.
\end{proof}

Given a basis for
an inner product space, the \textit{dual basis} is defined by the condition that the inner
products of a vector with the elements of the dual basis provide the
coefficients in the expansion of that vector in terms of the original
basis. In our case, let $\{\widetilde{E}_i\}$ denote the basis dual to $\{E_i\}$ so that, for any vector $A\in\mathcal{L}(\mathcal{H}_d)$,
\begin{equation}\label{dualdef}
    A=\sum_j(\tr A\widetilde{E}_j)E_j\;.
\end{equation}
One consequence of this definition is that if we expand the
original basis in terms of itself,
\begin{equation}
E_i = \sum_j (\tr E_i \widetilde{E}_j) E_j\;,
\end{equation}
linear independence of the $\{E_i\}$ implies that
\begin{equation}
    \tr E_i \widetilde{E}_j = \delta_{ij}\;,
\end{equation}
from which one may easily see that the original basis is the dual of the dual basis,
\begin{equation}
    A=\sum_j(\tr A E_j)\widetilde{E}_j\;.
\end{equation}
In the familiar case when the original basis is orthonormal,
the dual basis coincides with it: When we write a vector $\mathbf{v}$ as
an expansion over the unit vectors $(\mathbf{\hat{x}}, \mathbf{\hat{y}}, \mathbf{\hat{z}})$, the
coefficient of $\mathbf{\hat{x}}$ is simply the inner product of $\mathbf{\hat{x}}$ with
$\mathbf{v}$. 

A MIC is a positive semidefinite operator basis. For positive semidefinite operators $A$ and $B$, $\tr AB=0$ iff $AB=0$. Recall that a Hermitian matrix which is neither positive semidefinite nor negative semidefinite is known as an \textit{indefinite} matrix. 
\begin{theorem}
    The dual basis of a MIC is composed entirely of indefinite matrices.
\end{theorem}
\begin{proof}
    Suppose $\widetilde{E}_1\geq 0$. The definition of a dual basis tells us $\tr\widetilde{E}_1E_k=0$ for all $k\neq1$. Because they are both positive semidefinite, $\widetilde{E}_1E_k=0$ for all $k\neq1$. This means the $d^2-1$ MIC elements other than $E_1$ are operators on a $d-\text{rank}(\widetilde{E}_1)$ dimensional subspace. But 
\begin{equation}
    \text{dim}\left[\mathcal{L}\left(\mathcal{H}_{d-\text{rank}(\widetilde{E}_1)}\right)\right]\leq(d-1)^2<d^2-1,
\end{equation}
so they cannot be linearly independent. If $\widetilde{E}_1\leq0$, $-\widetilde{E}_1$ is positive semidefinite and the same logic holds.
\end{proof}
\begin{corollary}\label{noprops}
    No element in a MIC can be proportional to an element of the MIC's
    dual basis.
\end{corollary}
\begin{corollary}\label{noortho}
    No MIC can form an orthogonal basis.
\end{corollary}
\begin{proof}
    Suppose $\{E_i\}$ is a MIC which forms an orthogonal basis, that is, $\tr E_iE_j=c_i\delta_{ij}$ for some constants $c_i$. Summing this over $i$ reveals $c_j=e_j$, the weights of the MIC. Thus the dual basis is given by $\widetilde{E}_j=E_j/e_j=\rho_j$ which is a violation of Corollary \ref{noprops}.
\end{proof}
\begin{corollary}
    No MIC outcome can ever be assigned probability $1$.
\end{corollary}
\begin{proof}
    MIC probabilities provide the expansion coefficients for a state in the dual basis. If $P(E_i)=1$ for some $i$, the state would equal the dual basis element, but a state must be positive semidefinite.
\end{proof}
\begin{corollary}\label{noprojs}
    No effect of a MIC can be an unscaled projector.
\end{corollary}
\begin{proof}
    Suppose $E_1$ were equal to an unscaled projector $P$. Then any eigenvector of $P$ is a pure state which would imply probability 1 for the MIC outcome $E_1$, which is impossible. 
\end{proof}
Theorem~\ref{noprops} and the subsequent corollaries have physical meaning. In classical probability
theory, we grow accustomed to orthonormal bases. For example, imagine
an object that can be in any one of $N$ distinct configurations. When
we write a probability distribution over these $N$ alternatives, we
are encoding our expectations about which of these configurations is
physically present --- about the ``physical condition'' of the object,
as Einstein would say~\cite{Stacey:2018}, or in more modern
terminology, about the object's ``ontic
state''~\cite{Spekkens:2007}. We can learn everything there is to know
about the object by measuring its ``physical condition'', and any
implementation of such an ideal measurement is represented by
conditional probabilities that are 1 in a single entry and 0
elsewhere. In other words, the map from the object's physical
configuration to the reading on the measurement device is, at its
most complicated, a permutation of labels. Without loss of generality,
we can take the vectors that define the ideal measurement to be the
vertices of the probability simplex: The measurement basis is
identical with its dual, and the dual-basis elements simply label the
possible ``physical conditions'' of the object which the measurement
reads off.

In quantum theory, by contrast, no element of a MIC may be proportional to an element in the dual. This stymies the identification of the dual-basis elements as intrinsic ``physical conditions'' ready for a measurement to read.

\begin{theorem}\label{nopovm}
    No elementwise rescaling of a proper subset of a MIC may form a POVM. 
\end{theorem}
\begin{proof}
    Since a MIC is a linearly independent set, the identity element is uniquely formed by the defining expression
    \begin{equation}
        {I}=\sum_{i=1}^{d^2}E_i.
        \label{povmdef}
    \end{equation}
    If a linear combination of a proper subset $\Omega$ of the MIC elements could be made to also sum to the identity, 
\begin{equation}
    {I}=\sum_{i\in \Omega}\alpha_i E_i,
    \label{lincomb}
\end{equation}
then subtracting \eqref{lincomb} from \eqref{povmdef} implies
\begin{equation}
    0=\sum_{i\in\Omega}(1-\alpha_i)E_i+\sum_{i\notin\Omega}E_i
\end{equation}
which is a violation of linear independence.
\end{proof}
\begin{corollary}\label{d2ortho}
    No two elements in a $d=2$ MIC may be orthogonal under the Hilbert--Schmidt inner product.
\end{corollary}
\begin{proof}
    An orthogonal pair of elements in dimension $2$ may be rescaled such that they sum to the identity element. Therefore, by Theorem \ref{nopovm}, they cannot be elements of a MIC.
\end{proof}
These results also have physics
implications. For much of the history of quantum mechanics, one type
of POVM had special status: the \emph{von Neumann measurements,} which
consist of $d$ elements given by the projectors onto the vectors of an
orthonormal basis of $\mathbb{C}^d$. Indeed, in older books, these are
the only quantum measurements that are considered (often being defined
as the eigenbases of Hermitian operators called ``observables''). We
can now see that, from the standpoint of informational completeness,
the von Neumann measurements are rather pathological: There is no way
to build a MIC by augmenting a von Neumann measurement with additional
outcomes. 

Another holdover from the early days of quantum theory concerns the process of updating a quantum state in response to a measurement outcome. If one restricts attention to von Neumann measurements, one may feel tempted to grant special importance to the post-measurement state being one of the eigenvectors of an ``observable''. This type of updating is a special case of the more general theory developed as quantum mechanics was understood more fully. The \emph{L\"uders
Rule} \cite{Busch:2009,Barnum:2002}
states that the post-measurement state upon obtaining the outcome associated with effect $E_i$ for a POVM $\{E_i\}$ is 
\begin{equation}
    \rho_i':=\frac{\sqrt{E_i}\rho\sqrt{E_i}}{\tr \rho E_i}\;.
\end{equation}
In the special case of a von Neumann measurement, this reduces to replacing the state for the system with the eigenprojector corresponding to the measurement outcome. A physicist who plans to follow that procedure and then repeat the measurement immediately afterward would expect to obtain the same outcome twice in succession. Some authors regard this possibility as the essential point of contact with classical mechanics and attempt to build an
understanding of quantum theory around such ``ideal'' measurements~\cite{Cabello:2019}. But, as we said in the introduction, obtaining the same outcome twice in succession is not a good notion of a ``classical ideal''. Especially in view of the arbitrariness of von Neumann measurements from our perspective, we regard this possibility as conceptually downstream from the phenomenon of informationally complete measurements. 

Corollary \ref{d2ortho} prompts a question: May any elements of a MIC in arbitrary dimension be orthogonal? In other words, can any entry in a $G$ matrix equal zero? We answer this question in the affirmative with an explicit example of a rank-$1$ MIC in dimension $3$ with $7$ orthogonal pairs.
\begin{example}\label{7orthopairs}
    When multiplied by $1/3$, the following is a rank-1 unbiased MIC in dimension $3$ with $7$ orthogonal pairs.
    \renewcommand\arraystretch{1.2}
    \begin{equation}
        \begin{split}
            &\left\{\begin{bmatrix}
        1 & 0 & 0 \\
        0&0&0\\
        0&0&0
    \end{bmatrix},
    \begin{bmatrix}
        0 & 0 & 0 \\
        0&1&0\\
        0&0&0
    \end{bmatrix},
    \begin{bmatrix}
        \frac{1}{2} & 0 & \frac{1}{2} \\
        0&0&0\\
        \frac{1}{2}&0&\frac{1}{2}
    \end{bmatrix},
    \begin{bmatrix}
        0 & 0 & 0 \\
        0&\frac{1}{2}&\frac{1}{2}\\
        0&\frac{1}{2}&\frac{1}{2}
    \end{bmatrix},
    \begin{bmatrix}
        \frac{1}{2} & 0 & \frac{i}{2} \\
        0&0&0\\
        -\frac{i}{2}&0&\frac{1}{2}
    \end{bmatrix},
    \begin{bmatrix}
        0 & 0 & 0 \\
        0&\frac{1}{2}&\frac{i}{2}\\
        0&-\frac{i}{2}&\frac{1}{2}
    \end{bmatrix},
    \begin{bmatrix}
        \frac{1}{3} & \frac{i}{3} & -\frac{i}{3} \\
        -\frac{i}{3}&\frac{1}{3}&-\frac{1}{3}\\
        \frac{i}{3}&-\frac{1}{3}&\frac{1}{3}
    \end{bmatrix},\right.\\
    &\left. \begin{bmatrix}
        \frac{5}{8} & -\frac{1}{8}-\frac{i}{4} & -\frac{3}{8}-\frac{i}{8} \\
        -\frac{1}{8}+\frac{i}{4}&\frac{1}{8}&\frac{1}{8}-\frac{i}{8}\\
        -\frac{3}{8}+\frac{i}{8}&\frac{1}{8}+\frac{i}{8}&\frac{1}{4}
    \end{bmatrix},
    \begin{bmatrix}
        \frac{1}{24} & \frac{1}{8}-\frac{i}{12} & -\frac{1}{8}-\frac{i}{24} \\
        \frac{1}{8}+\frac{i}{12}&\frac{13}{24}&-\frac{7}{24}-\frac{3i}{8}\\
        -\frac{1}{8}+\frac{i}{24}&-\frac{7}{24}+\frac{3i}{8}&\frac{5}{12}
    \end{bmatrix}\right\}.
\end{split}
\end{equation}
These are projectors onto the following vectors in $\mathcal{H}_d$:
\begin{equation}
    \begin{split}
        &\left\{(1,0,0),(0,1,0),\frac{1}{\sqrt{2}}(1,0,1),\frac{1}{2}(0,1,1),\frac{1}{\sqrt{2}}(1,0,-i),\frac{1}{\sqrt{2}}(0,1,-i),\frac{1}{\sqrt{3}}(1,-i,i),\right.\\
        &\left.\frac{1}{\sqrt{40}}(5,-1+2i,-3+i),\frac{1}{\sqrt{24}}(1,3+2i,-3+i)\right\}.
\end{split}
\end{equation}
The Gram matrix of the MIC elements is
\renewcommand\arraystretch{1.2}
\begin{equation}
\begin{bmatrix}
    \frac{1}{9}&0&\frac{1}{18}&0&\frac{1}{18}&0&\frac{1}{27}&\frac{5}{72}&\frac{1}{216}\\
    0&\frac{1}{9}&0&\frac{1}{18}&0&\frac{1}{18}&\frac{1}{27}&\frac{1}{72}&\frac{13}{216}\\

    \frac{1}{18}&0&\frac{1}{9}&\frac{1}{36}&\frac{1}{18}&\frac{1}{36}&\frac{1}{27}&\frac{1}{144}&\frac{5}{432}\\
    0&\frac{1}{18}&\frac{1}{36}&\frac{1}{9}&\frac{1}{36}&\frac{1}{18}&0&\frac{5}{144}&\frac{1}{48}\\
    \frac{1}{18}&0&\frac{1}{18}&\frac{1}{36}&\frac{1}{9}&\frac{1}{36}&0&\frac{5}{144}&\frac{1}{48}\\
    0&\frac{1}{18}&\frac{1}{36}&\frac{1}{18}&\frac{1}{36}&\frac{1}{9}&\frac{1}{27}&\frac{1}{144}&\frac{5}{432}\\
\frac{1}{27}&\frac{1}{27}&\frac{1}{27}&0&0&\frac{1}{27}&\frac{1}{9}&\frac{1}{54}&\frac{1}{18}\\
\frac{5}{72}&\frac{1}{72}&\frac{1}{144}&\frac{5}{144}&\frac{5}{144}&\frac{1}{144}&\frac{1}{54}&\frac{1}{9}&\frac{1}{27}\\
\frac{1}{216}&\frac{13}{216}&\frac{5}{432}&\frac{1}{48}&\frac{1}{48}&\frac{5}{432}&\frac{1}{18}&\frac{1}{27}&\frac{1}{9}\\
\end{bmatrix}.
\end{equation}
\end{example}
\noindent The numerical search resulting in this example led us to formulate the following:
\begin{conjecture}
A rank-1 MIC in dimension 3 can have no more than 7 pairs of
orthogonal elements.
\end{conjecture}

Our next result characterizes when it is possible to build a rank-1
POVM out of a set of vectors and specifies the additional 
conditions which must be met in order for it to form a MIC. We make use of the \textit{Hadamard product}~\cite{Horn:1994}, denoted $\circ$, which is elementwise multiplication of matrices.

\begin{theorem}\label{r1povmtightframe}
    Consider a set of $N$ normalized vectors $\ket{\phi_i}$ in $\mathcal{H}_d$ and real numbers $0 \leq e_i\leq 1$. The following are equivalent:
    \begin{enumerate}
        \item $E_i:= e_i\ketbra{\phi_i}{\phi_i}$ forms a rank-1 POVM.
        \item The Gram matrix $g$ of the rescaled vectors $\sqrt{e_i}\ket{\phi_i}$ is a rank-$d$ projector.
    \end{enumerate}
    Furthermore, if $N=d^2$ and ${\rm rank}(g\circ g^*)=d^2$, $\{E_i\}$ forms a rank-1 MIC.
\end{theorem}
\begin{proof}
    Suppose $E_i$ forms a rank-1 POVM, that is, 
    \begin{equation}\label{rank1povm}
        \sum_ie_i\ketbra{\phi_i}{\phi_i}=I\;.
    \end{equation}
    It is easy to see that this is only possible if the set $\{\sqrt{e_i}\ket{\phi_i}\}$ spans $\mathcal{H}_d$, and, consequently, $N\geq d$. It now follows that $g$ is a rank-$d$ projector because the left hand side of \eqref{rank1povm} is a matrix that has the same nonzero spectrum as $g$~\cite{Waldron:2018}. On the other hand, if $g$ is a rank-$d$ projector, $N\geq d$ and
    $\{\sqrt{e_i}\ket{\phi_i}\}$ spans a $d$ dimensional space because the rank of a Gram matrix is equal to the dimension of the space spanned by the vectors. Using again the fact the left hand side of \eqref{rank1povm} has the same nonzero spectrum as the Gram matrix, it must equal the identity and thus the rank-1 POVM condition holds.

 To be a MIC, $N$ must equal $d^2$. The remaining condition on $\{E_i\}$ for it to form a rank-1 MIC is that its elements be linearly independent. This is equivalent to the condition that its Gram matrix $G$ is full rank. The relation between $g$ and $G$ is given by the Hadamard product of $g$ with its conjugate,
\begin{equation}
    g \circ g^*=G\;,
\end{equation}
and so, if $N=d^2$ and ${\rm rank}(g\circ g^*)=d^2$, $\{E_i\}$ forms a rank-1 MIC. 
\end{proof}
\noindent For any two matrices $A$ and $B$, the Hadamard product satisfies the
rank inequality
\begin{equation}\label{rank}
    \text{rank}(A\circ B) \leq \text{rank}(A)\;\text{rank}(B)\;,
\end{equation}
so a rank-1 MIC is produced when $\text{rank}(g\circ g^*)$ achieves its maximal value with the minimal number of effects. Perhaps this criterion will lead to a way to conceptualize rank-1 MICs directly in terms of the vectors in $\mathcal{H}_d$ from which they can be constructed. 

As a brief illustration, all rank-$d$ projectors are unitarily equivalent so the specification of $g$ for the rescaled vectors of any rank-1 MIC is obtainable from any rank-$d$ projector by conjugating it with the right unitary.
Specifying $g$ specifies the MIC: $g$ is equal to its own square root, so its columns are the vectors up to unitary equivalence which form this Gram matrix. To obtain these vectors as elements of $\mathcal{H}_d$, one can simply write them in the basis provided by the eigenvectors of $g$ with nonzero eigenvalues. From a fixed starting projector, then, finding a rank-1 MIC is equivalent to choosing a unitary in $U(d^2)$ which maximizes $\text{rank}(g\circ g^*)$. Numerically this
maximization appears to be typical, but we are not aware of an explicit characterization. A further question to
ask is whether there are special classes of unitaries which give particular types of MICs. 

We finish this section with a very brief discussion of the geometry of MIC space. For this purpose it is not necessary to distinguish between MICs which differ only in permutations of their effects. We further discuss this in the loose sense of not having chosen any particular metric. The full sets of $N$-outcome POVMs are in general convex manifolds~\cite{DAriano:2005}, but the requirement of linear independence prevents this from being true for MICs --- it is possible for a convex mixture of MICs to introduce a
linear dependence and thus step outside of the set. There do, however, exist infinite sequences and curves lying entirely within the set of MICs. In these terms one can see that the space of MICs lacks much of its boundary, that is, one can construct infinite sequences of MICs for which the limit point is not a MIC. The simplest such limit point is the POVM consisting of the identity and $d^2-1$ zero matrices. Similarly there are MICs arbitrarily close to any POVM with fewer than
$d^2$ elements which has been padded by zero matrices. Among unbiased MICs, another limit point lying outside of the set is the trivial POVM consisting of $d^2$ identical matrices $E_i=\frac{1}{d^2}I$. Provided they exist, SICs are limit points, at least among equiangular MICs (see section \ref{equiangularMICs}), which are contained within the set.

\section{Explicit Constructions of MICs}
\label{sec:constructions}
\subsection{SICs}
The MICs that have attracted the most interest are the SICs, which in
many ways are the optimal MICs~\cite{Appleby:2014b, Appleby:2015,
  DeBrota:2018, Fuchs:2003, Scott:2006}.  SICs were studied as
mathematical objects (under the name ``complex equiangular lines'')
before their importance for quantum information was
recognized~\cite{Delsarte:1975, Hoggar:1981, Coxeter:1991,
  Hoggar:1998}.  Prior to SICs becoming a physics problem,
constructions were known for dimensions $d = 2$, 3 and
8~\cite{Konig:2001}. Exact solutions for SICs are now known in 79
dimensions:
\begin{equation}
    \begin{split}
    d = 2&\hbox{--}28,30,31,35,37\hbox{--}39,42,43,48,49,52,53,57,61\hbox{--}63,67,73,74,78,79,84,91,93,\\
    & 95,97\hbox{--}99,103,109,111,120,124,127,129,134,143,146,147,168,172,195,199,\\
    & 228,259,292,323,327,399,489,844,1299.
\end{split}
\end{equation}
The expressions for these solutions grow complicated quickly, but
there is hope that they can be substantially
simplified~\cite{Appleby:2018}. Numerical solutions have also been
extracted, to high precision, in the following dimensions:
\begin{equation}
    d= 2\hbox{--}193,204,224,255,288,528,725,1155,2208.
\end{equation}
Both the numerical and the exact solutions have been found in
irregular order and by various methods. Many entries in these lists
are due to A.\ J.\ Scott and M.\ Grassl~\cite{Scott:2010a, Scott:2017,
  Grassl:2017}; other explorers in this territory include M.\ Appleby,
I.\ Bengtsson, T.-Y.\ Chien, S.\ T.\ Flammmia, G.\ S.\ Kopp and
S.\ Waldron.

Together, these results have created the community sentiment that SICs
\emph{should} exist for every finite value of~$d$. To date, however, a
general proof is lacking. The current frontier of SIC research extends
into algebraic number theory~\cite{Appleby:2013, Appleby:2016,
  Bengtsson:2016, Appleby:2017b, Kopp:2018}, which among other things
has led to a method for uplifting numerical solutions to exact
ones~\cite{Appleby:2017}. The topic has begun to enter the textbooks
for physicists~\cite{Bengtsson:2017} and for
mathematicians~\cite{Waldron:2018}.

The effects of a SIC are given by 
\begin{equation}
    E_i = \frac{1}{d}\Pi_i, \hbox{ where } \Pi_i = \ketbra{\pi_i}{\pi_i}\;,
\end{equation}
where we will take the liberty of calling any of the sets $\{E_i\}$, $\{\Pi_i\}$, and $\{\ket{\pi_i}\}$ SICs.
It is difficult to find a meaningful visualization of structures in
high-dimensional complex vector space. However, for the $d = 2$ case,
an image is available. Any quantum state for a 2-dimensional system
can be written as an expansion over the Pauli matrices:
\begin{equation}
  \rho = \frac{1}{2}\left(I + x\sigma_x + y\sigma_y +
  z\sigma_z\right).
\end{equation}
The coefficients $(x,y,z)$ are then the coordinates for $\rho$ in the
\emph{Bloch ball}. The surface of this ball, the \emph{Bloch sphere,}
lives at radius 1 and is the set of pure states. In this picture, the
quantum states $\{\Pi_i\}$ comprising a SIC form a regular
tetrahedron; for example,
\begin{equation}
  \Pi_{s,s'} = \frac{1}{2}\left(I + \frac{1}{\sqrt{3}}\left(s\sigma_x + s'\sigma_y +  ss'\sigma_z\right)\right),
\end{equation}
where $s$ and $s'$ take the values $\pm 1$.

The matrix $G_{\rm SIC}$ has the spectrum
\begin{equation}
  \lambda(G_{\rm SIC}) = \left(
  \frac{1}{d}, \frac{1}{d(d+1)}, \ldots, \frac{1}{d(d+1)}
  \right).
\end{equation}
The flatness of this spectrum will turn out to be significant; we will
investigate this point in depth in the next section.


\subsection{MICs from Random Bases}\label{sec:micfrombases}
It is possible to construct a MIC for any dimension $d$. Let $\{A_i\}$ be any basis of positive semidefinite operators in $\mathcal{L}(\mathcal{H}_d)$ and define $\Omega:=\sum_i A_i$. Then 
\begin{equation}\label{arbitraryMIC}
    E_i:=\Omega^{-1/2}A_i\Omega^{-1/2}
\end{equation}
forms a MIC. If $\{A_i\}$ consists entirely of rank-1 matrices, we obtain a rank-1 MIC.\footnote{In the rank-1 case, this procedure is equivalent to forming what is called the \textit{canonical tight frame} associated with the frame of vectors in $\mathcal{H}_d$ whose outer products form the $A_i$ matrices. For more information on this, see \cite{Waldron:2018}.} If $\{A_i\}$ is already a MIC, $\Omega=I$ and the transformation is trivial; MICs are the fixed points of this mapping from one positive
semidefinite operator basis to another. Thanks to this property, this method can produce \textit{any} MIC if the initial basis is drawn from the full space of positive semidefinite
operators. 


This procedure was used by Caves, Fuchs and Schack in the course of
proving a quantum version of the de Finetti
theorem~\cite{Caves:2002c}. (For background on this theorem, a key
result in probability theory, see~\cite[\S 5.3]{Stacey:2015}
and~\cite{Diaconis:2017}.) We refer to the particular MICs they constructed as the
\emph{orthocross MICs.} As the orthocross MICs are of historical importance, we explicitly detail their construction and provide some first properties and conjectures about it in the remainder of this subsection.

To construct an orthocross MIC in dimension
$d$, first pick an orthonormal basis $\{\ket{j}\}$. This is a set of
$d$ objects, and we want a set of $d^2$, so our first step is to take
all possible combinations:
\begin{equation}
\Gamma_{jk} := \ketbra{j}{k}.
\end{equation}
The orthocross MIC will be built from a set of $d^2$ rank-1 projectors
$\{\Pi_\alpha\}$, the first $d$ of which are given by
\begin{equation}
\Pi_\alpha = \Gamma_{\alpha\alpha}.
\end{equation}
Then, for $\alpha = d+1, \ldots, \frac{1}{2}d(d+1)$, we take all the
quantities of the form
\begin{equation}
\frac{1}{2}\left(\ket{j} + \ket{k}\right)
 \left(\bra{j} + \bra{k}\right)
= \frac{1}{2} (\Gamma_{jj} + \Gamma_{kk} + \Gamma_{jk} + \Gamma_{kj}),
\end{equation}
where $j < k$. We construct the rest of the $\{\Pi_\alpha\}$
similarly, by taking all quantities of the form
\begin{equation}
\frac{1}{2}\left(\ket{j} + i\ket{k}\right)
 \left(\bra{j} - i\bra{k}\right)
= \frac{1}{2} (\Gamma_{jj} + \Gamma_{kk} - i\Gamma_{jk} + i\Gamma_{kj}),
\end{equation}
where again the indices satisfy $j < k$. That is, the set
$\{\Pi_\alpha\}$ contains the projectors onto the original orthonormal
basis, as well as projectors built from the ``cross terms''.

The operators $\{\Pi_\alpha\}$ form a positive semidefinite operator basis which can be plugged into the procedure described above. Explicitly,
\begin{equation}
\Omega = \sum_{\alpha=1}^{d^2} \Pi_\alpha\;,
\end{equation}
and the orthocross MIC elements are given by
\begin{equation}
    E_\alpha := \Omega^{-1/2} \Pi_\alpha \Omega^{-1/2}\;.
\end{equation}

The operator $\Omega$ for the initial set of vectors has a comparatively simple matrix
representation: The elements along the diagonal are all equal to~$d$,
the elements above the diagonal are all equal to $\frac{1}{2}(1-i)$,
and the rest are $\frac{1}{2}(1+i)$, as required by $\Omega =
\Omega^\dag$. The matrix $\Omega$ is not quite a circulant matrix,
thanks to that change of sign, but it can be turned into one by
conjugating with a diagonal unitary matrix. Consequently, the
eigenvalues of~$\Omega$ can be found explicitly via discrete Fourier
transformation.  The result is that, for $m = 0,\ldots,d-1$,
\begin{equation}
\lambda_m = d + \frac{1}{2}\left(\cot \frac{\pi(4m+1)}{4d} - 1 \right).
\end{equation}
This mathematical result has a physical
implication~\cite{Fuchs:2002}.
\begin{theorem}
  The probability of any outcome $E_\alpha$ of an orthocross MIC,
  given any quantum state $\rho$, is bounded above by
  \begin{equation}
    P(E_\alpha) \leq \left[d - \frac{1}{2}\left(1 + \cot
      \frac{3\pi}{4d}\right)\right]^{-1} < 1.
  \end{equation}
\end{theorem}
\begin{proof}
  The maximum of $\tr(\rho E_\alpha)$ over all $\rho$ is bounded above
  by the maximum of $\tr(\Pi E_\alpha)$, where $\Pi$ ranges over the
  rank-1 projectors. In turn, this is bounded above by the maximum
  eigenvalue of~$E_\alpha$. We then invoke that
  \begin{equation}
    \lambda_{\rm max}(E_\alpha)
    = \lambda_{\rm max}(\Omega^{-1/2}\Pi_\alpha \Omega^{-1/2})
    = \lambda_{\rm max}(\Pi_\alpha \Omega^{-1} \Pi_\alpha)
    \leq \lambda_{\rm max} (\Omega^{-1}).
  \end{equation}
  The desired bound then follows.
\end{proof}

Note that all the entries in the matrix $2\Omega$ are Gaussian integers,
that is, numbers whose real and imaginary parts are
integers. Consequently, all the coefficients in the characteristic
polynomial of~$2\Omega$ will be Gaussian integers, and so the eigenvalues
of~$2\Omega$ will be roots of a monic polynomial with Gaussian-integer
coefficients. This is an example of how, in the study of MICs, number
theory becomes relevant to physically meaningful quantities --- in
this case, a bound on the maximum probability of a
reference-measurement outcome. Number theory has also turned out to be
very important for SICs, in a much more sophisticated
way~\cite{Appleby:2013, Appleby:2016, Bengtsson:2016, Appleby:2017b,
  Kopp:2018}.

The following conjectures about orthocross MICs have been motivated by
numerical investigations. We suspect that their proofs will be relatively straightforward, but so far they have eluded us.
\begin{conjecture}
The entries in $G$ for orthocross MICs can become
arbitrarily small with increasing $d$, but no two elements of an
orthocross MIC can be exactly orthogonal.
\end{conjecture}
\begin{conjecture}
For any orthocross MIC, the entries in $G^{-1}$ are integers or half-integers.
\end{conjecture}

\subsection{Group Covariant MICs}\label{sec:groupcovariant}
The method discussed in the previous subsection allows us to make fully arbitrary MICs, but it is also possible to construct MICs with much more built-in structure. The MICs which have received the most attention in the literature to date are the \textit{group covariant} MICs --- those whose elements are the orbit of a group of unitary matrices acting by conjugation. For additional discussion of group covariant IC POVMs, see \cite{DAriano:2004}. 

The Gram matrix of a group covariant MIC is very simple. Suppose $\{E_i\}$ is a group covariant MIC, so $E_i=U_iE_0U_i^\dag$ where $E_0$ is the first element of the MIC and the index $i$ gives the element of the unitary representation of the group sending this element to the $i$th element. Then all distinct elements of the Gram matrix are present in the first row because
\begin{equation}
    [G]_{ij}=\tr E_iE_j=\tr U_iE_0U_i^\dag U_jE_0U_j^\dag=\tr E_0U_kE_0U_k^\dag\;,
\end{equation}
for some $k$ determined by the group. Another way to say this is that every row of the Gram matrix of a group covariant MIC is some permutation of the first row. 

Note that any group covariant MIC is unbiased because conjugation by a unitary cannot change the trace of a matrix, but the converse is not true; the simplest example of an unbiased MIC which is not group covariant which we have encountered is the one given in Example \ref{7orthopairs}. 

The most studied and likely most important group covariant MICs are the \textit{Weyl--Heisenberg} MICs (WH MICs), which are covariant with respect to the Weyl--Heisenberg group. Part of the intution for the importance of this group comes from the fact that its generators form finite dimensional analogs of the position and momentum operators. Perhaps more telling, every known SIC is group covariant and in all cases but one that group is the
Weyl--Heisenberg group~\cite{Fuchs:2017a}. 

The Weyl--Heisenberg group is constructed as follows. Let $\{\ket{j}: j = 0,\ldots,d-1\}$ be an
orthonormal basis, and define $\omega = e^{2\pi i/d}$. Then the
operator
\begin{equation}
X\ket{j} = \ket{j+1},
\end{equation}
where addition is interpreted modulo $d$, effects a cyclic shift of
the basis vectors. The Fourier transform of the $X$ operator is
\begin{equation}
Z \ket{j} = \omega^j \ket{j},
\end{equation}
and together these operators satisfy the Weyl commutation relation
\begin{equation}
ZX = \omega XZ.
\end{equation}
The Weyl--Heisenberg displacement operators are
\begin{equation}
D_{k,l} := (-e^{\pi i/d})^{kl} X^k Z^l,
\end{equation}
and together they satisfy the conditions
\begin{equation}
D_{k,l}^\dag = D_{-k,-l},\quad\ D_{k,l}D_{m,n} = (-e^{\pi i/d})^{lm-kn}
D_{k+m,l+n}.
\end{equation}
Each $D_{k,l}$ is unitary and a $d^{\rm th}$ root of the identity. The
Weyl--Heisenberg group is the set of all operators $(-e^{\pi i/d})^m
D_{k,l}$ for arbitrary integers $m$, and it is projectively equivalent
to $\mathbb{Z}_d \times \mathbb{Z}_d$. Then, for any density matrix $\rho$ such that
\begin{equation}
    \tr \!\!\left(D_{k,l}^\dag\rho\right)\neq 0,\quad \forall (k,l)\in \mathbb{Z}_d\times\mathbb{Z}_d\;,
\end{equation}
the set
\begin{equation}
    E_{k,l}:=\frac{1}{d}D_{k,l}\rho D_{k,l}^\dag
\end{equation}
forms a WH MIC. 
\subsection{Equiangular MICs}\label{equiangularMICs}
An \textit{equiangular}\footnote{Appleby and Graydon introduced the term \emph{SIM} for an equiangular MIC of arbitrary rank; a rank-1 SIM is a
SIC~\cite{Graydon:2016, Graydon:2016a}. 
}
 MIC is one for which the Gram matrix takes the form
\begin{equation}
    [G]_{ij}=\alpha\delta_{ij}+\zeta\;.
\end{equation}
Equiangular MICs are unbiased (see Corollary 3 in \cite{Appleby:2015}) and, because $\sum_{ij}[G]_{ij}=d$, it is easy to see that $\alpha=1/d-d^2\zeta$ and that   
\begin{equation}
    \frac{1}{d^2(d+1)}\leq\zeta<\frac{1}{d^3}\;.
\end{equation}
SICs are rank-1 equiangular MICs for which $\zeta$ achieves the minimum allowed value. The upper bound $\zeta$ value is approached by MICs arbitrarily close to $E_i=\frac{1}{d^2}I$ for all $i$.

Armed with a SIC in a given dimension, one can construct an equiangular MIC for any allowed $\zeta$ value by mixing in some of the identity to each element: 
\begin{equation}\label{remixedSIC}
    E_i=\frac{\beta}{d}\Pi_i+\frac{1-\beta}{d^2}I\;,\quad \frac{-1}{d-1}\leq(\beta\neq0)\leq1\;.
\end{equation}
Even if a SIC is not known, it is generally much easier to construct equiangular MICs when the elements are not required to be rank-1. One way to do this which always works for any $\beta\leq\frac{1}{d+1}$ is by replacing the SIC projector in equation \eqref{remixedSIC} with a quasi-SIC.\footnote{See Appendix A of \cite{DeBrota:2019} for the definition and a construction of a quasi-SIC.} Depending on the quasi-SIC, higher values of $\beta$ may also work.

Another construction in odd dimensions are the \emph{Appleby
MICs}~\cite{Appleby:2007a}. The Appleby MICs are WH covariant and constructed as follows. Let the operator $B$ be constructed as
\begin{equation}
    B := \frac{1}{\sqrt{d+1}} \sum_{\{k,l\}\neq\{0,0\}} D_{k,l},
\end{equation}
and define $B_{k,l}$ to be its conjugate under a Weyl--Heisenberg
displacement operator:
\begin{equation}
B_{k,l} := D_{k,l} B D_{k,l}^\dag.
\end{equation}
The elements of the Appleby MIC have rank $(d+1)/2$, and are defined by
\begin{equation}
E_{k,l} := \frac{1}{d^2} \left(I + \frac{1}{\sqrt{d+1}} B_{k,l}\right).
\end{equation}

For any quantum state $\rho$, the quantities
\begin{equation}
W_{k,l} := (d+1) \tr(E_{k,l} \rho) - \frac{1}{d}
\end{equation}
are \emph{quasiprobabilities}: They can be negative, but the sum over
all of them is unity. The quasiprobability function $\{W_{k,l}\}$ is
known as the \emph{Wigner function} of the quantum state $\rho$. This is an example of a relation we study much more generally in a companion paper~\cite{DeBrota:2019b}. In particular, we speculate there that the Appleby MIC and other MICs may inherit special significance from a related a Wigner function.

\subsection{Tensorhedron MICs}
So far, we have not imposed any additional structure upon our Hilbert
space. However, in practical applications, one might have additional
structure in mind, such as a preferred factorization into a tensor
product of smaller Hilbert spaces. For example, a register in a
quantum computer might be a set of $N$ physically separate qubits,
yielding a joint Hilbert space of dimension $d = 2^N$. In such a case,
a natural course of action is to construct a MIC for the joint system
by taking the tensor product of multiple copies of a MIC defined on
the component system:
\begin{equation}
  E_{j_1,j_2,\ldots,j_N} := E_{j_1} \otimes E_{j_2} \otimes \cdots \otimes E_{j_N}.
\end{equation}
Since a collection of $N$ qubits is a natural type of system to
consider for quantum computation, we define the $N$-qubit
\emph{tensorhedron MIC} to be the tensor product of $N$ individual
qubit SICs. These have appeared in the setting of quantum cryptography \cite{Englert:2005} as well as quantum tomography \cite{Zhu:2010} where it was proven that tensor product SICs (and thus $N$-qubit tensorhedron MICs) are optimal among product measurementss in the same way that SICs are across all measurements. 

\begin{theorem}
  The Gram matrix of an $N$-qubit tensorhedron MIC is the tensor
  product of $N$ copies of the Gram matrix for the qubit SIC out of
  which the tensorhedron is constructed.
\end{theorem}
\begin{proof}
  Consider the two-qubit tensorhedron MIC, whose elements are given by
  \begin{equation}
    E_{d(j-1)+j'} := \frac{1}{4}\Pi_j \otimes \Pi_{j'},
  \end{equation}
  with $\{\Pi_j\}$ being a qubit SIC. The Gram matrix for the
  tensorhedron MIC has entries
  \begin{equation}
    [G]_{d(j-1)+j',d(k-1)+k'} = \frac{1}{16}
    \tr[(\Pi_j\otimes\Pi_{j'}) (\Pi_k\otimes\Pi_{k'})].
  \end{equation}
  We can group together the projectors that act on the same subspace:
  \begin{equation}
    [G]_{d(j-1)+j',d(k-1)+k'} = \frac{1}{16}
    \tr(\Pi_j \Pi_k \otimes \Pi_{j'}\Pi_{k'}).
  \end{equation}
  Now, we distribute the trace over the tensor product, obtaining 
  \begin{equation}
    [G]_{d(j-1)+j',d(k-1)+k'} = \frac{1}{16}
    \frac{2\delta_{jk}+1}{3} \frac{2\delta_{j'k'} + 1}{3}
    = [G_{\rm SIC}]_{jk} [G_{\rm SIC}]_{j'k'},
  \end{equation}
  which is just the definition of the tensor product:
  \begin{equation}
    G = G_{\rm SIC} \otimes G_{\rm SIC}.
  \end{equation}
  This extends in the same fashion to more qubits.
\end{proof}
\begin{corollary}
  The spectrum of the Gram matrix for an $N$-qubit tensorhedron MIC
  contains only the values
  \begin{equation}
    \lambda = \frac{1}{2^N}\frac{1}{3^m},\ m = 0,\ldots,N.
  \end{equation}
\end{corollary}
\begin{proof}
  This follows readily from the linear-algebra fact that the spectrum
  of a tensor product is the set of products $\{\lambda_i \mu_j\}$,
  where $\{\lambda_i\}$ and $\{\mu_j\}$ are the spectra of the
  factors.
\end{proof}
We can also deduce properties of MICs made by taking tensor products
of MICs that have orthogonal elements.  Let $\{E_j\}$ be a
$d$-dimensional MIC with Gram matrix $G$, and suppose that exactly $N$
elements of $G$ are equal to zero.  The tensor products $\{E_j \otimes
E_{j'}\}$ construct a $d^2$-dimensional MIC, the entries in whose Gram
matrix have the form $[G]_{jk} [G]_{j'k'}$, as above.  This product
will equal zero when either factor does, meaning that the Gram matrix
of the tensor-product MIC will contain $2d^4 N - N^2$ zero-valued
entries.  It seems plausible that in prime dimensions, where
tensor-product MICs cannot exist, the possible number of zeros is more
tightly bounded, but this remains unexplored territory.

\section{SICs are Minimally Nonclassical Reference Measurements}
\label{sec:optimal}
What might it mean for a MIC to be the best among all MICs? Naturally, it depends on what qualities are valued in light of which one MIC may be superior to another. As mentioned in the introduction, for a large number of metrics, SICs are optimal. The authors of this paper particularly value the capacity of MICs to index probabilistic representations of the Born Rule. For this use, the best MIC is the one which provides the most useful probabilistic representation, adopting some
quantitative ideal that a representation should approach. One codification of such an ideal is as follows. In essence, we want to find a MIC that
furnishes a probabilistic representation of quantum theory which looks
as close to classical probability as is mathematically possible. The
residuum that remains --- the unavoidable discrepancy that even the
most clever choice of MIC cannot eliminate --- is a signal of what is
truly \emph{quantum} about quantum mechanics. 

In a recent paper it was shown that SICs are strongly optimal for this project~\cite{DeBrota:2018}. To see why, consider the following scenario. An agent has a physical system
of interest, and she plans to carry out either one of two different,
mutually exclusive procedures on it.  In the first
procedure, she will drop the system directly into a measuring apparatus
and thereby obtain an outcome.  In the second procedure, she will
cascade her measurements, sending the system through a reference
measurement and then, in the next stage, feeding it into the device
from the first procedure.  Probability theory unadorned by physical assumptions
provides no constraints binding her expectations
for these two different courses of action. Let $P$ denote her probability assignments for the
consequences of following the two-step procedure and $Q$ those for the
single-step procedure.  Then, writing $\{H_i\}$ for the possible
outcomes of the reference measurement and $\{D_j\}$ for those of the
other,
\begin{equation}
  P(D_j) = \sum_i P(H_i) P(D_j|H_i).
\end{equation}
This equation is a consequence
of Dutch-book coherence~\cite{Fuchs:2013, Diaconis:2017} known
as the Law of Total Probability (LTP). But the claim that
\begin{equation}
  Q(D_j) = P(D_j)
\end{equation}
is an assertion of \emph{physics,} not entailed by the rules of probability theory alone. This assertion codifies in probabilistic language
the classical ideal that a reference measurement simply reads off the
system's ``physical condition'' or ``ontic state''.

We know this classical ideal is not met in quantum theory, that is, $Q(D_j)\neq P(D_j)$. Instead, as detailed in reference \cite{DeBrota:2018}, $Q(D_j)$ is related to $P(H_i)$ and $P(D_j|H_i)$ in a different way. To write the necessary equations compactly, we introduce a vector notation where the LTP takes the form
\begin{equation}
    P(D) = P(D|H)P(H)\;.
\end{equation}
To set up the quantum version of the above scenario, let $\{H_i\}$ be a MIC and $\{D_j\}$ be an arbitrary POVM. Furthermore, let $\{\sigma_i\}$ denote a set of post-measurement states for the reference measurement; that is, if the agent experiences outcome $H_i$, her new state for the system will be $\sigma_i$. In this notation, the Born Rule becomes
\begin{equation}\label{ltpanalog}
  Q(D) = P(D|H)\Phi P(H)\;, \hbox{ with }
  [\Phi^{-1}]_{ij} := \tr H_i \sigma_j\;.
\end{equation}
The matrix $\Phi$ depends upon the MIC and the post-measurement states, but it is always a column
quasistochastic matrix, meaning its columns sum to one but may contain
negative elements~\cite{DeBrota:2018}. In fact, $\Phi$ \emph{must}
contain negative entries; this follows from basic structural
properties of quantum theory~\cite{Ferrie:2011}. Now, the classical intution we mentioned above would be expressed by $\Phi=I$. However, no choice of MIC and set of post-measurement states can achieve this. The MICs and post-measurement sets which give a $\Phi$ matrix closest to the identity therefore supply the ideal representation we seek.

Theorem 1 in reference \cite{DeBrota:2018} proves that the distance between $\Phi$ and the identity with respect to any unitarily invariant norm is minimized when both the MIC and the post-measurement states are proportional to a SIC. Unitarily invariant norms include the Frobenius norm, the trace norm, the operator norm, and all
the other Schatten $p$-norms, as well as the Ky Fan $k$-norms. Although this theorem was proven for foundational reasons, a special case of the result turns out to
answer in the affirmative a conjecture regarding a practical matter of
quantum computation~\cite[\S VII.A]{Veitia:2018}.

What ended up being important for the optimality proof in \cite{DeBrota:2018} was that both the MIC and the post-measurement states be proportional to SICs, but not necessarily that they be proportional to the \textit{same} SIC. Although the measures considered there were not sensitive to this distinction, the same SIC case has obvious conceptual and mathematical advantages. From a conceptual standpoint, when the post-measurement states are simply the projectors $\Pi_i$ corresponding to the SIC outcome just
obtained, our ``throw away and reprepare'' process is equivalent to L\"uders rule updating, which there are independent reasons for preferring~\cite{Barnum:2002}. When the post-measurement states are the same SIC as the reference measurement, $\Phi$ takes the uniquely simple form
\begin{equation}
    \Phi_{\rm SIC}=(d+1)I-\frac{1}{d}J\;,
\end{equation}
where $J$ is the Hadamard identity, that is, the matrix of all 1s. Inserted into \eqref{ltpanalog} and written in index form, this produces the expression
\begin{equation}
    Q(D_j)=\sum_i\left[(d+1)P(H_i)-\frac{1}{d}\right]P(D_j|H_i)\;,
    \label{urgleichung}
\end{equation}
having the advantage that for each conditional probability given an outcome $H_i$, only the $i^{\rm th}$ reference probability figures into that term in the sum. This is not so for two arbitrarily chosen SICs, and, as such, that case would result in a messier probabilistic representation.

This path is not the only one from which to arrive at the conclusion that SICs furnish a minimally nonclassical reference measurement. Recall the close association of classicality and orthogonality noted in section~\ref{sec:basics}. From this standpoint, one might claim that most ``classical'' or least ``quantum'' reference measurement is one that is closest to an orthogonal measurement. 

While we know from Corollary \ref{noortho} that a MIC cannot be an orthogonal basis, how close can one get? One way to quantify this closeness is via an operator distance between the Gramians of an orthogonal basis and a MIC. From the proof of Corollary \ref{noortho}, we know that if a MIC could be orthogonal its Gram matrix would be $[G]_{ij}=e_i\delta_{ij}$. With no further restrictions, we can get arbitrarily close to this ideal, for instance, with a MIC constructed as follows.
Consider a set of $d^2$ matrices $\{A_i\}$ where the first $d$ of them are the eigenprojectors of a Hermitian matrix and the remaining $d^2-d$ are the zero matrix. Then, for an arbitrary\footnote{As long as a linear dependence does not develop.} MIC $\{B_j\}$, we may form a new MIC, indexed by a real number $0<t<1$,
\begin{equation}
    E^t_i:=t A_i+(1-t)B_i\;.
\end{equation}
One may see that the Gram matrix of $\{E^t_i\}$ approaches the orthogonal Gram matrix in the limit $t\rightarrow 1$.

But at such an extreme, the usefulness of a MIC is completely destroyed. In the above scenario when $t$ is close to $1$, the informational completeness is all but gone, as one has to reckon with vanishingly small probabilities when dealing with a MIC close to the limit point. Such a MIC fails miserably at being anything like a reasonable reference measurement. Although formally capable of being a reference measurement, a biased MIC deprives us of an even-handed treatment of indifference; the
garbage state, which is poised in Hilbert space to capture pure state preparation indifference, would be represented by a non-flat probability distribution. Worse, for any sufficiently biased MIC, i.e., one with any weight less than $1/d^2$, the flat probability distribution is not reached by any density matrix. Consequently, what we're really after is an unbiased reference measurement which is as close to an orthogonal measurement as possible. With this additional constraint,
the following theorem demonstrates that SICs are the optimal choice.

\begin{theorem}\label{closetoortho}
    The closest an unbaised MIC can be to an orthogonal basis, as
    measured by the Frobenius distance between their Gramians, is when
    the MIC is a SIC.\footnote{An earlier paper by one of us (CAF) and
      a collaborator~\cite{Fuchs:2013} made the claim that the
      condition of being unbiased could be derived by minimizing the
      squared Frobenius distance; this is erroneous as the unequally
      weighted example with $\{E^t_i\}$ shows. For the purposes of
      that earlier paper, it is sufficient to impose by hand the
      requirement that the MIC be unbiased, since this is a naturally
      desirable property for a standard reference measurement. Having
      made this extra proviso, the conceptual conclusions of that work
      are unchanged.}
\end{theorem}

\begin{proof}
    We lower bound the square of the Frobenius distance:
    \begin{equation}
        \begin{split}
        \sum_{ij}\left(\frac{1}{d}\delta_{ij}-\tr E_iE_j\right)^{\!2}&=\sum_i\left(\frac{1}{d}-\tr E_i^2\right)^{\!2}+\sum_{i\neq j}(\tr E_iE_j)^2\\
        &\geq \frac{1}{d^2}\left(\sum_i\left(\frac{1}{d}-\tr E_i^2\right)\right)^{\!2}+\frac{1}{d^4-d^2}\left(\sum_{i\neq j}\tr E_iE_j\right)^{\!2}\\
        &=\frac{1}{d^2}\left(d-\sum_i\tr E_i^2\right)^{\!2}+\frac{1}{d^4-d^2}\left(d-\sum_i\tr E_i^2\right)^{\!2}\\
        &=\frac{1}{d^2-1}\left(d-\sum_i\tr E_i^2\right)^{\!2}\geq\frac{(d-1)^2}{d^2-1}=\frac{d-1}{d+1}\;.
    \end{split}
    \end{equation}
    The first inequality follows from two invocations of the Cauchy--Schwarz inequality and achieves equality iff $\tr E_i^2$ and $\tr E_iE_j$, for $i\neq j$, are constants, that is, iff the MIC is an equiangular MIC. The third line is easy to derive from the fact that for any MIC, $\sum_{ij}[G]_{ij}=d$. The final inequality comes from noting that $\sum_i\tr E_i^2=\frac{1}{d^2}\sum_i\tr\rho_i^2\leq1$ with equality iff the MIC is rank-1. Thus the lower bound is
    saturated iff the equal weight MIC is rank-1 and equiangular, that is, iff it is a SIC.
\end{proof}
Theorem \ref{closetoortho} concerned the Gramian of a MIC. We can, in fact, show a stronger result on the inverse of the Gram matrix. 
\begin{theorem}\label{U-norm}
    Let $G$ be the Gram matrix of an unbiased MIC, and let
    $\norm{\cdot}$ be any unitarily invariant norm (i.e., any norm
    where $\norm{A} = \norm{UAV}$ for arbitrary unitaries $U$ and
    $V$). Then
    \begin{equation}
        \left\|{I-\frac{1}{d}G^{-1}}\right\|\geq \left\|{I - \frac{1}{d}G_{\rm SIC}^{-1}}\right\|\;,
    \end{equation}
    with equality if and only if the MIC is a SIC.
\end{theorem}
\begin{proof}
This is a special case of Theorem 1 in~\cite{DeBrota:2018}.
\end{proof}

As with the theorems we proved above about MICs in general, this mathematical result has physical meaning. Classically speaking, the ``ideal of the detached observer'' (as
Pauli phrased it~\cite{Fuchs:2017b}) is a measurement that reads off
the system's point in phase space, call it $\lambda_i$, without
disturbance. A state of maximal certainty is one where an agent is absolutely certain which $\lambda_i$ exists.  An agent having maximal certainty about each of a
pair of identically prepared systems implies that she expects to obtain
the same outcome for a reference measurement on each system. In other words, her ``collision
probability'' is unity:
\begin{equation}
\sum_i p(\lambda_i)^2 = 1.
\label{eq:cprob-condition}
\end{equation}
There is also a quantum condition on states of maximal certainty. As before, we can approach the question, ``What is the unavoidable residuum that separates quantum from classical?''\ by finding the form of this quantum condition that brings it as close as possible to the classical version.

\begin{lemma}
  Given a MIC $\{E_i\}$ with Gramian $G$, a quantum state is pure if
  and only if its probabilistic representation satisfies
  \begin{equation}
    \sum_{ij} p(E_i) p(E_j) [G^{-1}]_{ij} = 1.
  \end{equation}
\end{lemma}
\begin{proof}
    Let $\{E_i\}$ be a MIC. The expansion of any quantum state $\rho$ in the dual basis is
\begin{equation}
\rho = \sum_i (\tr E_i \rho) \widetilde{E}_i\;.
\end{equation}
By the Born Rule, the coefficients are probabilities:
\begin{equation}
    \rho = \sum_i p(E_i) \widetilde{E}_i\;.
\end{equation}
Now, recall that while $\tr \rho = 1$ holds for any quantum state, $\tr
\rho^2 = 1$ holds if and only if that operator is a pure state,
i.e., a rank-1 projector.  These operators are the extreme points of
quantum state space; all other quantum states are convex combinations
of them. In terms of the MIC's dual basis, the pure-state condition is
\begin{equation}
    \sum_{ij} p(E_i)p(E_j) \tr \widetilde{E}_i \widetilde{E}_j = 1\;,
\end{equation}
and so, because the Gramian of the dual basis is the inverse of the MIC Gram matrix,
\begin{equation}
\sum_{ij} p(E_i) p(E_j) [G^{-1}]_{ij} = 1\;,
\label{eq:QM-condition}
\end{equation}
as desired.
\end{proof}

Equation \eqref{eq:QM-condition} closely resembles the collision probability, \eqref{eq:cprob-condition}. If $G^{-1}$ were the identity, they would be identical. On the face of it, it looks as though we should see how close $G^{-1}$ can get to the identity. One minor wrinkle is that we should actually compare $G^{-1}$ with $dI$ instead of just with $I$, because an unbiased, orthogonal MIC (if one could exist) would have the Gram matrix $\frac{1}{d}I$. So, how close can we bring $G^{-1}$ to $dI$, by choosing an
appropriate unbiased MIC?  We know the answer to this from Theorem \ref{U-norm}: The best choice is a SIC.

\section{Computational Overview of MIC Gramians}
\label{sec:numerics}
In order to explore the realm of MICs more broadly, and to connect
them with other areas of mathematical interest, it is worthwhile to
generate MICs \emph{randomly} and study the typical properties which result. In this section we focus on the Gram matrix spectra of four MIC varieties whose constructions are described in section \ref{sec:constructions}. These types are: 
\begin{enumerate}
    \item Generic MICs: a MIC generated from an arbitrary positive semidefinite basis
    \item Generic Rank-1 MICs: a MIC generated from an arbitrary rank-1 positive semidefinite basis 
    \item WH MICs: a MIC obtained from the WH orbit of an arbitrary density matrix
    \item Rank-1 WH MICs: a MIC obtained from the WH orbit of an arbitrary pure state density matrix.
\end{enumerate}
In Hilbert space dimensions $2$ through $5$ we generated $10^5$ MICs with the following methodologies. We constructed the generic MICs as in section \ref{sec:micfrombases} and the WH MICs as in section \ref{sec:groupcovariant}. Each generic MIC was obtained from a basis of positive semidefinite operators and each WH MIC was obtained from the orbit of an initial density matrix. In the generic rank-1 case, the pure states defining the basis of projectors were sampled uniformly from the Haar
measure. Likewise, in the rank-1 WH case, the initial vector was also sampled uniformly from the Haar measure. The positive semidefinite bases for the arbitrary-rank generic MICs and the initial states for the arbitrary-rank WH MICs were constructed as follows. First, Hermitian matrices $M$ were sampled from the Gaussian Unitary distribution, and, for each of these, the positive semidefinite matrix $M^\dag M$ was formed. $d^2$ of these sufficed to form a positive semidefinite
basis without loss of generality and a trace-normalized instance served as the initial state for the WH MICs. For each MIC, we constructed its Gram matrix and computed the eigenvalues. Figures \ref{fig:d2}, \ref{fig:d3}, \ref{fig:d4}, and
\ref{fig:d5} are histograms of the eigenvalue distributions for dimensions $2$, $3$, $4$, and $5$, respectively. 

We note some expected and unexpected features of these distributions. In accordance with Theorem \ref{unbiased}, both group covariant types, being unbiased, always have the maximal eigenvalue $1/d$, while this is the lower bound for the maximal eigenvalue for the other two types. Particularly in the unbiased cases, because the eigenvalues must sum to $1$, not all of them can be too large, so it is perhaps not surprising that there are few eigenvalues approaching $1/d$ and that all families show exponential
decay until that value. However, the spectra of rank-1 MICs, especially in dimensions 2 and 3 (Figures \ref{fig:d2} and \ref{fig:d3}), display a richness of features for which we have no explanation. 

Most surprising of all is the small eigenvalue plateau in Figure \ref{fig:d3} for the $d=3$ rank-1 WH MICs. Further scrutiny has revealed that the plateau ends precisely at $1/12$, the average value for the non-maximal eigenvalues of an unbiased $d=3$ MIC Gram matrix. The Gram matrix for a $d=3$ SIC has the
spectrum
\begin{equation}
    \left(\frac{1}{3},\frac{1}{12},\frac{1}{12},\frac{1}{12},\frac{1}{12},\frac{1}{12},\frac{1}{12},\frac{1}{12},\frac{1}{12}\right)\;,
\end{equation}
which has the maximal amount of degeneracy allowed. Dimension $3$ is also exceptional in the study of SICs: It is the only known dimension for which there is a continuous family of unitarily inequivalent SICs~\cite{Tabia:2013, Hughston:2016}. Because the Gram matrix spectra for $d=3$ rank-1 WH MICs also behaves unlike the other dimensions we have checked and because the plateau appears to be connected with the value $1/12$, we conjecture that the eigenvalue plateau and the continuous family of SICs
may be related.
\begin{conjecture}
The plateau in the eigenvalue distribution for $d = 3$, seen in
Figure~\ref{fig:d3}, is related to the existence of a
continuous family of unitarily inequivalent SICs in that
dimension.
\end{conjecture}

\begin{figure}[b]
    \centering
    \includegraphics[width=\linewidth]{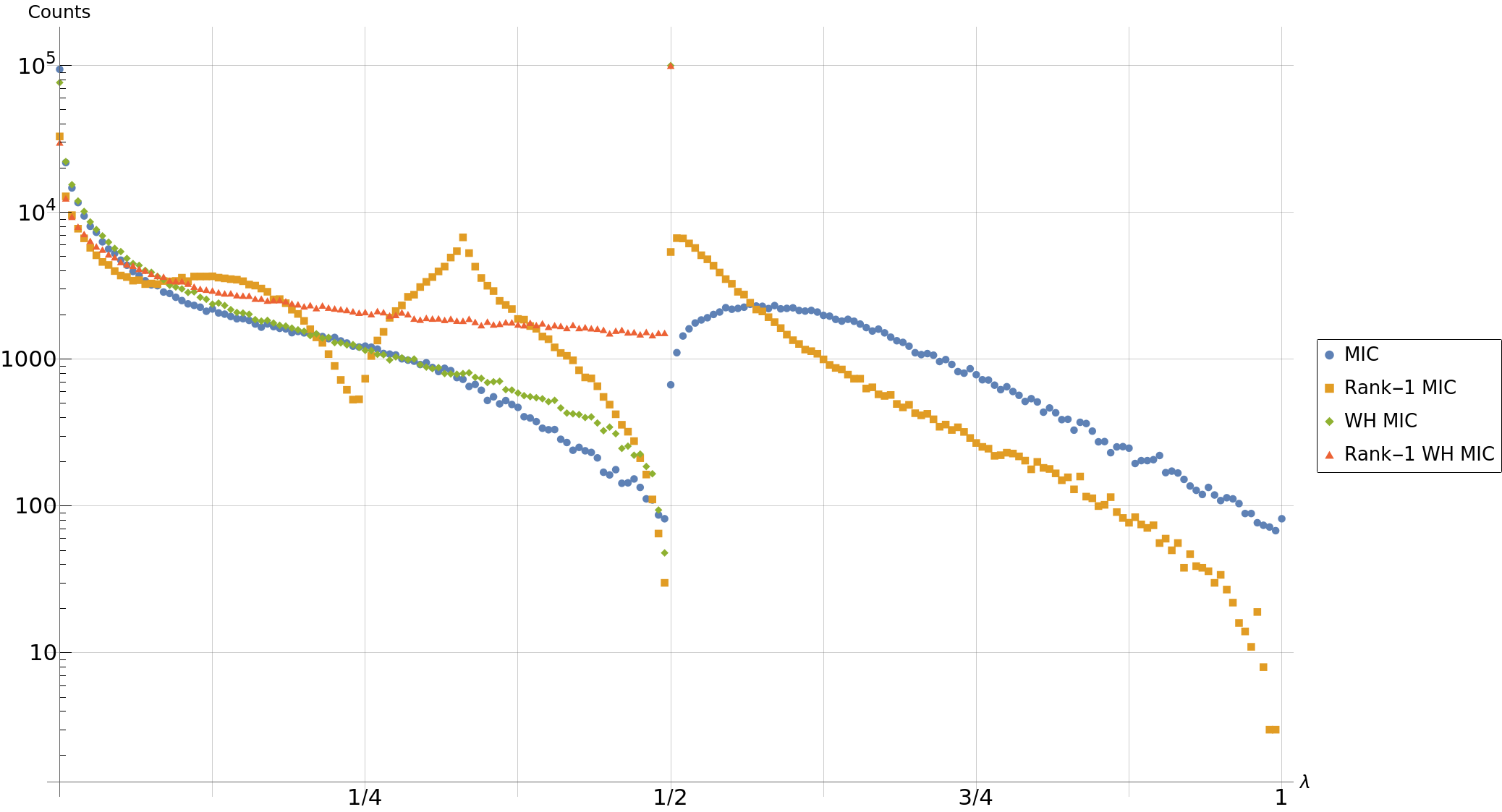}
    \caption{$d=2$ random MIC Gram matrix spectra, $N=10^5$, bin size $1/200$.}
    \label{fig:d2}
\end{figure}

\begin{figure}[b]
    \centering
    \includegraphics[width=\linewidth]{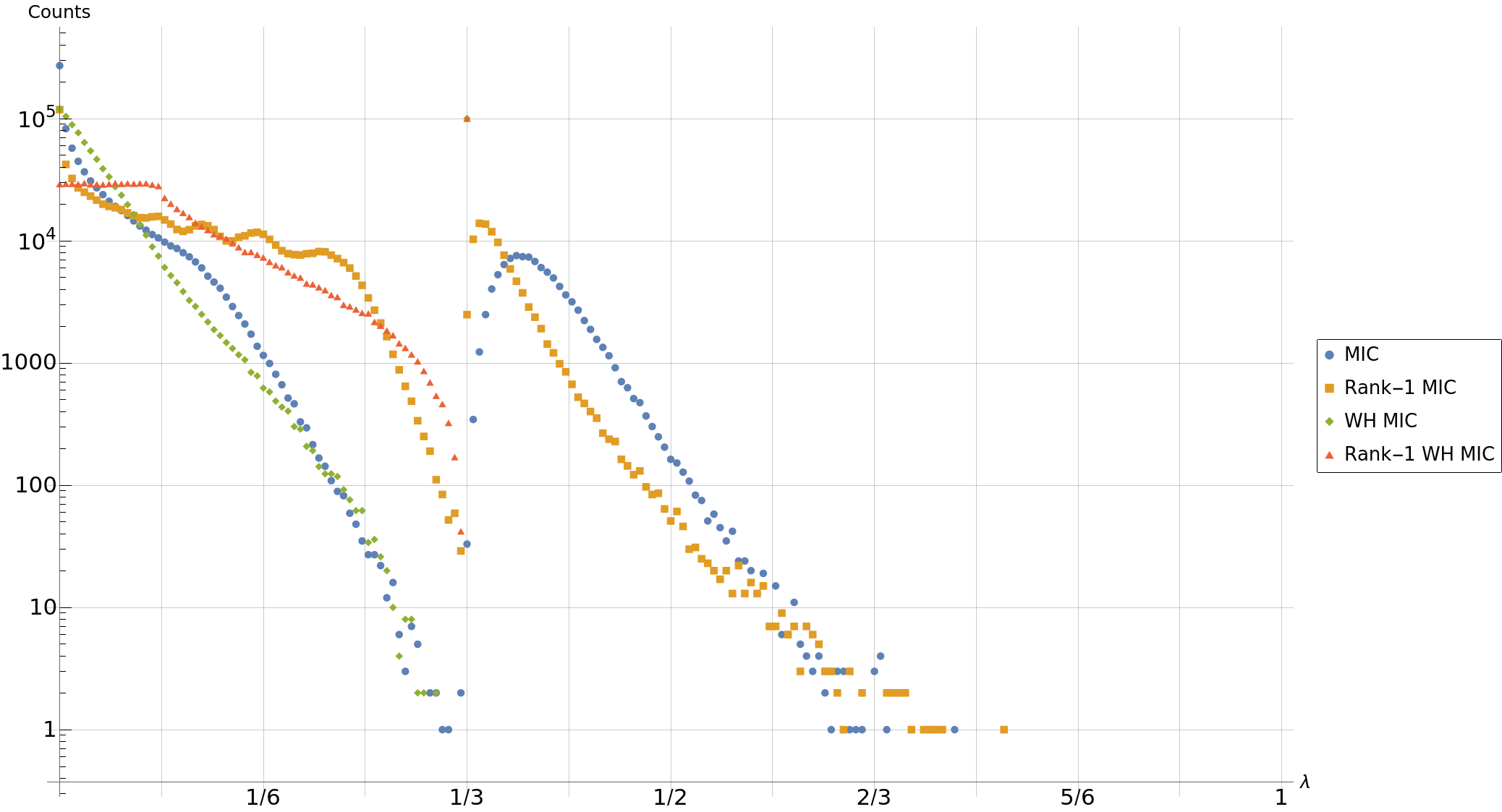}
    \caption{$d=3$ random MIC Gram matrix spectra, $N=10^5$, bin size $1/198$.}
    \label{fig:d3}
\end{figure}

\begin{figure}[b]
    \centering
    \includegraphics[width=\linewidth]{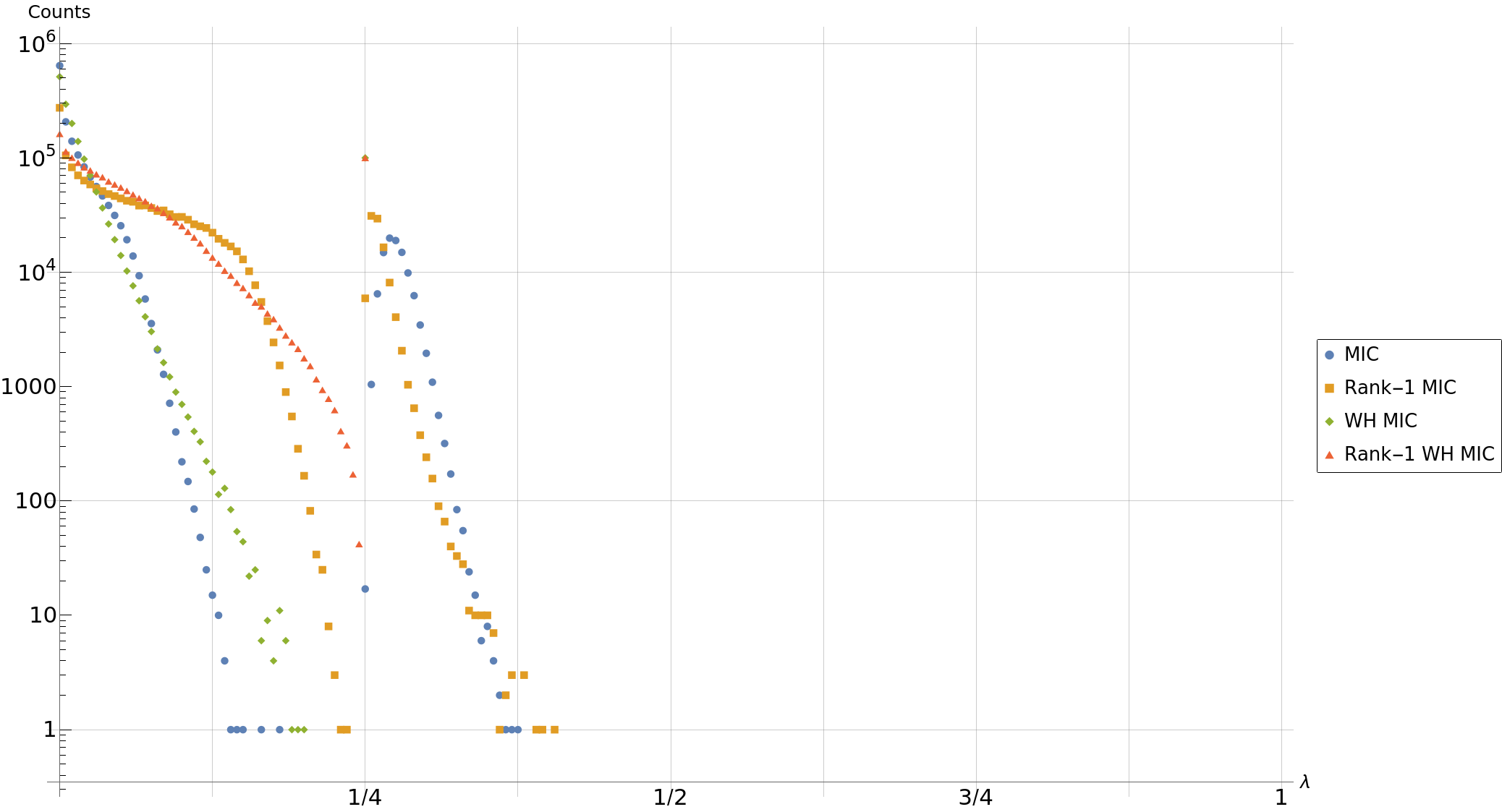}
    \caption{$d=4$ random MIC Gram matrix spectra, $N=10^5$, bin size $1/200$.}
    \label{fig:d4}
\end{figure}

\begin{figure}[b]
    \centering
    \includegraphics[width=\linewidth]{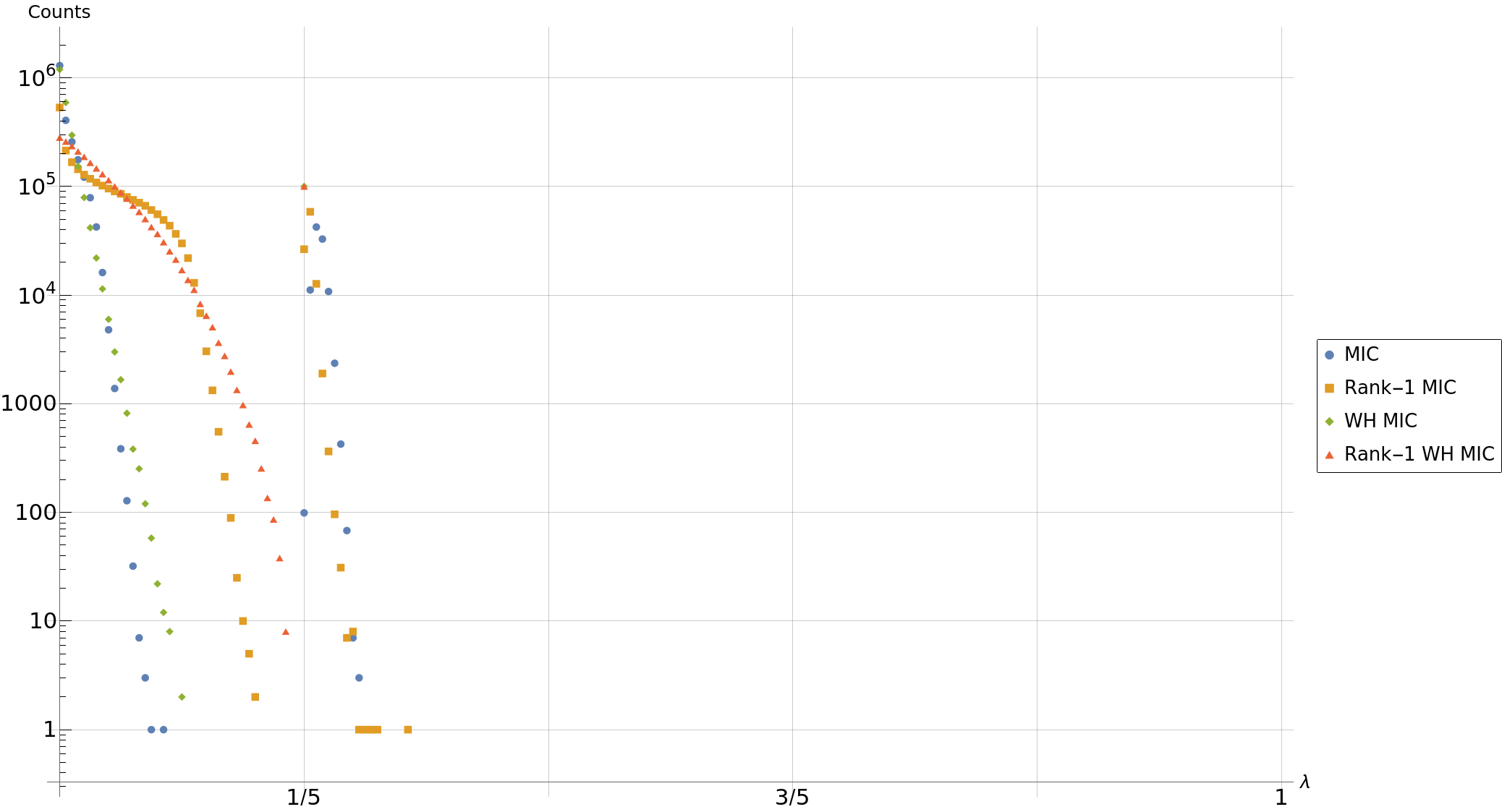}
    \caption{$d=5$ random MIC Gram matrix spectra, $N=10^5$, bin size $1/200$.}
    \label{fig:d5}
\end{figure}
%
%
%
%
%
%

\section{Conclusions}

We have argued that informational completeness provides the right
perspective from which to compare the quantum and the classical. The
structure of Minimal Informationally Complete quantum measurements and
especially how and to what degree this structure requires the
abandonment of classical intuitions therefore deserves explicit
study. We have surveyed the domain of MICs and derived some initial
results regarding their departure from such classical intuitions as
orthogonality, repeatability, and the possibility of
certainty. Central to understanding MICs are their Gram matrices; it
is through properties of these matrices that we were able to derive
many of our results. We have only just scratched the surface of this
topic, as our conjectures and unexplained numerical features of Gram
matrix spectra can attest.  In a sequel, we will explore another
application of Gram matrices. They hold a central role in the
construction of \emph{Wigner functions} from MICs~\cite{Stacey:2016c,
  Zhu:2016a, DeBrota:2017}, and Wigner functions are a topic pertinent
to quantum computation~\cite{Wootters:1987, Gibbons:2004, Veitch:2012,
  Veitch:2014, Howard:2016}.

Many properties of MIC Gram matrices remain unknown. Numerical
investigations have, in some cases, outstripped the proving of
theorems, resulting in the conjectures we have enumerated. Another
avenue for potential future exploration is the application of Shannon
theory to MICs. Importing the notions of information theory into
quantum mechanics has proved quite useful over the years at
illuminating strange or surprising features of the
physics~\cite{Braunstein:1990, Fuchs:1998, Buck:2000}. One promising
avenue of inquiry is studying the probabilistic representations of
quantum states using entropic measures. In the case of SICs, this has
already yielded intriguing connections among information theory, group
theory and geometry~\cite{Stacey:2016c, Stacey:2016b,
  Slomczynski:2014, Szymusiak:2014, Szymusiak:2015}. The analogous
questions for other classes of MICs remain open for investigation.

\section*{Acknowledgments}
We thank Marcus Appleby, Lane Hughston, Peter
Johnson, Steven van Enk and Huangjun Zhu for discussions. This research was supported in part by the John E.\ Fetzer Memorial
Trust, the John Templeton Foundation, and grants FQXi-RFP-1612 and
FQXi-RFP-1811B of the Foundational Questions Institute and Fetzer
Franklin Fund, a donor advised fund of Silicon Valley Community
Foundation.  The opinions expressed in this publication are those of
the authors and do not necessarily reflect the views of the John
Templeton Foundation.

\bibliographystyle{utphys}

\providecommand{\href}[2]{#2}\begingroup\raggedright\endgroup

\end{document}